\documentclass[journal, onecolumn]{IEEEtran}
\usepackage{caption}
\usepackage{subcaption}
\usepackage{graphicx}	
\usepackage{amssymb}
\usepackage[cmex10]{amsmath}
\usepackage{fixltx2e}
\usepackage{hyperref}
\usepackage{setspace}
\usepackage{comment}
\usepackage{eqparbox}
\usepackage{url}
\usepackage{cite}
\usepackage{color}
\usepackage{multirow}
\usepackage[belowskip=2pt,aboveskip=1pt]{subcaption}
\usepackage{array}
 \usepackage{algorithmic}
\usepackage{algorithm}
\usepackage[utf8]{inputenc}
\usepackage{amsthm}

\newtheorem{lem}{Lemma}
\newtheorem{prop}{Proposition}
\usepackage{enumitem}

\begin{document}

\title{Joint MOO of Transmit Precoding and Receiver Design in a  Downlink Time Switching MISO SWIPT System }

\author{Nafiseh Janatian, 
        Ivan~Stupia 
        and~Luc~Vandendorpe
}


\maketitle

\begin{abstract}
In this paper, we consider a time-switching (TS) co-located simultaneous wireless information and power transfer (SWIPT) system consisting of multiple multi-antenna access points which serve multiple single antenna users. In this scenario, we propose a multi-objective optimization (MOO) framework to design jointly the Pareto optimal beamforming vector and the TS ratio for each receiver. The objective is to maximize the utility vector including the achieved data rates and the harvested energies of all users simultaneously. This problem is a non-convex rank-constrained MOO problem which is relaxed and transformed into a non-convex semidefinite program (SDP) based on  the weighted Chebycheff method. The majorization-minimization algorithm is utilized  to solve the nonconvex SDP and the optimal solution is proved to satisfy the rank constraint. We also study the problem of optimizing the beamforming vectors in a fixed TS ratio scenario with the same approach. Numerical results are provided for \textcolor{black}{two coordinated access points with MISO configuration}. The results illustrate the trade-off between harvested energy and information data rate objectives and show the effect of optimizing the precoding strategy and TS ratio on this trade-off. 
\end{abstract}


%
\IEEEpeerreviewmaketitle

\section{Introduction}
With the increasing concern about the \textcolor{black}{integration of energy constrained ultra low power devices into the future wireless ecosystem}, a growing attention has been recently  devoted to the concept of RF energy harvesting.  The idea of using the same electromagnetic field for transferring both information and power to the wireless devices, called simultaneous wireless information and power transfer (SWIPT), is one of the most appealing techniques in this context. 
SWIPT is a promising solution to increase the lifetime of wireless nodes 
and hence alleviate the energy bottleneck of energy constrained wireless networks.
It is predicted that SWIPT will become an indispensable building block for many commercial and industry wireless systems in the future, including the upcoming 
internet of things (IoT) systems, wireless sensor networks and small-cell networks \cite{bi}. The ideal SWIPT receiver architecture assumes that energy can be extracted from the same signal as that used for information decoding \cite{varsh}. However, the current circuit designs are not yet able to implement this extraction, since the energy carried by the RF signal is lost during the information decoding process.
As a result,  a considerable effort has been devoted to the study of different practical SWIPT receiver architectures, namely, the parallel receiver architecture and the co-located 
 receiver architecture \cite{zhang}. A parallel receiver architecture equips the energy harvester and the information receiver with independent antennas for energy harvesting (EH) and information decoding (ID). In a co-located receiver architecture, the energy harvester and the information receiver share the same antennas. Two common methods to design such kind of receivers are time-switching (TS) and power-splitting (PS). In TS, the receiver switches in time between EH and ID, while in PS the receiver splits the received signal into two streams of different power values for EH and ID.

SWIPT has to be realized by properly allocating the available resources and sharing them among both information transfer and energy transfer. Designing TS/PS SWIPT receivers in a point-to-point wireless  environment to achieve various trade-offs between wireless information transfer and energy harvesting is considered in \cite{liu1,liu2,zho}. In multi-user environments, researches on SWIPT focus on the power and subcarrier allocation among different users such that some criteria (throughput, harvested power, fairness, etc.) are met. 
For the multi-user downlink channel, various policies have been proposed for single input-single output (SISO) and multi input-single output (MISO) configurations \cite{ng},\cite{shi}.
Resource allocation algorithm design aiming at the maximization of data transmission energy efficiency in a SISO PS SWIPT multi-user system is considered in \cite{ng}  with an orthogonal frequency division multiple access (OFDMA). A MISO configuration offers the additional degree of freedom of beamforming vector optimization at the transmitter.
In \cite{shi}, a joint beamforming and PS ratio allocation scheme was designed to minimize the power cost under the constraints of throughput and harvested energy. 
The problem of joint power control and time switching  in MISO SWIPT systems by considering the long-term power consumption and heterogeneous quality of service (QoS) requirements for different types of traffics is also studied in \cite{dong}. 
The energy efficient beamforming design in MISO heterogeneous cellular networks including separate EH and ID receivers is addressed in \cite{Sheng}. Beamformers are designed with two aims of maximizing the information transmission efficiency of ID users and energy harvesting efficiency  of EH users while taking the minimal rate requirement and the minimal harvested power threshold for ID users and EH users into account. 

SWIPT in a MIMO interference channel with two transmitter-receiver pairs is studied in \cite{park}. When both receivers are set in ID and EH mode, beamforming vectors are found to maximize the achievable sum rate and the harvested energy, respectively. Also for the mixed case of one ID receiver and one EH receiver, transmit strategies are proposed in order to maximize the energy transfer to the EH receiver and minimize the interference to the ID receiver. 
PS SWIPT in a multi-user MIMO interference channel scenario is also studied in \cite{zong}. The objective is to minimize the total transmit power of all transmitters by jointly designing the transmit beamformers, power splitters, and receiver filters, subject to the signal-to-interference-plus-noise ratio (SINR) constraint for ID and the harvested power constraint for EH at each receiver.

All the above works consider single cell cases with one base station (BS) and single or multiple mobile users. In a multi-cell case the system becomes interference limited. However, while interference links are harmful for information decoding, they are useful for energy harvesting. 
The SWIPT beamforming design for multiple cells with coordinated multipoint approach (CoMP) is addressed in \cite{ng2}. The objective is the minimization of the total network transmit power and the maximum capacity consumption for the backhaul link under constraints on the SINRs and the values of harvested energy. Sparse beamforming for real-time energy trading in CoMP-SWIPT networks is also studied in \cite{Nur}.

As we have seen so far in the literature overview of SWIPT, these works have considered single objective optimization (SOO) framework to formulate the problem of resource allocation or beamforming optimization. Popular objectives are classical performance metrics such as (weighted) sum rate/ throughput (to be maximized), or transmit power (to be minimized), or sum of energy harvested (to be maximized). In SOO one of these objectives is selected as the sole objective while the others are considered as constraints. This approach assumes that one of the objectives is of dominating importance and also it requires prior knowledge about the accepted values of the constraints related to the other objectives.  However, the multi-objective optimization (MOO)  investigates the optimization of the vector of objectives, for nontrivial situations where there is a conflict between objectives. This approach has been proposed lately for wireless information systems   \cite{bjornson} and is  considered for a parallel SWIPT system in \cite{moo, moo2} very recently. The considered system  consists of a multi-antenna transmitter, a single-antenna information receiver, and multiple energy harvesting receivers equipped with multiple antennas. In this scenario, the trade-off between the maximization of the energy efficiency of information transmission and the maximization of the wireless power transfer efficiency is studied by means of resource allocation using an MOO framework.

The fundamental approach used in this paper is also MOO. We address the joint MOO of transmit precoding and receiver design in a TS MISO SWIPT system and our main contributions are summarized as follows:

\begin{itemize}[leftmargin=*]
\item 
Different from previous works, we consider a MISO SWIPT system consisting of multiple multi-antenna access points (APs) which serve multiple single antenna user equipments (UEs). Therefore, the problem formulation and the proposed algorithm in this work can be simply applied to \textcolor{black}{a scenario in which distributed APs cooperate phase coherently via a X-haul network to simultaneously serve heterogeneous UEs. Heterogeneity includes different types of receivers like pure data receivers (e.g. smartphones and laptops), energy harvesters and wireless sensors that are capable of both harvesting the energy and decoding the information}. As a result, the term UE in this paper refers to a broader range of devices encompassing the ones directly used by the end-users and the autonomous sensors.
 
\item
We have assumed the TS SWIPT technique, which is practically feasible and can be implemented using simple switches, while PS receivers require highly complex hardware due to different power sensitivity values of ID and EH parts in each receiver. 
\textcolor{black}{In this perspective, it is worth mentioning that TS SWIPT receivers can be considered as a special case of dynamic PS SWIPT receivers with on-off power splitting factor. Hence, since realistic values of the ID and EH receivers sensitivity may differ by more than $30$dB, TS and dynamic PS SWIPT will have similar performance in practical scenarios.} 

\item 
We formulate the problem to design the optimal transmit precoding matrix and the time switching ratio of each receiver jointly to maximize the utility vector including the achieved information data rates and harvested energies of all users simultaneously using MOO framework. Since an MOO problem cannot be solved in a globally optimal way, the Pareto optimality   \cite{pareto} of the resource allocation will be adopted as optimality criterion. 
This problem is a non-convex rank-constrained MOO problem. First we relax the rank constraint and transform the problem into a non-convex SOO semidefinite program (SDP) and after that we show that the optimal solution satisfies the rank constraint. In this framework we utilize the majorization-minimization algorithm \cite{MM} to solve the nonconvex SOO SDP. Numerical results illustrate the trade-off between energy harvested and information data rate objectives  and show the effect of optimizing the precoding strategy and TS ratio on this trade-off. 
\end{itemize}

The rest of this paper is organized as follows. Section II describes the system model and problem formulation. Joint multi-objective design of spatial precoding and receiver time switching is studied in section III. In Section IV, we present numerical results and finally the paper is concluded in Section V.
\section{System Model and Problem Formulation}
We consider a multi-user MISO downlink system for SWIPT over flat fading channels as shown in Fig. \ref{fig1}. The system consists of $N_{AP}$ APs which are equipped with $N_{A_j}, j=1,..,N_{AP}$ antennas  and serve $N_{UE}$ single antenna UEs.  
The set of all UEs and all APs are denoted by $\mathcal{N}_{UE}$ and $\mathcal{N}_{AP}$, respectively.
Each user is assumed to be served by multiple transmitters but the information symbols will be coded and emitted independently. Therefore, the received signal in the $i$th UE can be modelled as:
\begin{equation} \label{M}
  y_i=\sum_{j=1}^{N_{AP}}\boldsymbol{h_{ij}^H}{\sum_{l=1}^{N_{UE}}\boldsymbol{x_{lj}}}s_{l}+ n_i,
\end{equation}
where $i,l\in{\mathcal{N}_{UE}}, j\in{\mathcal{N}_{AP}}$, $s_{l}$ is the information symbol from the APs to the $l$th UE which originates from independent Gaussian codebooks, ${s_{l}}\sim{\mathcal{CN}(\boldsymbol{0},1)}$ and  $\boldsymbol{x_{lj}}\in\mathbb{C}^{N_{A_j}\times1}$ is the beamforming vector. We assume quasi-static flat fading channels for all UEs and denote by  $\boldsymbol{h_{ij}}\in\mathbb{C}^{N_{A_j}\times1}$ the complex channel vector from the $j$th AP to the $i$th UE. Also $n_i\sim{\mathcal{CN}(0,\sigma_i^2)}$ is the circularly symmetric complex Gaussian receiver noise which includes  the antenna noise 
and the ID processing noise 
in the $i$th user.
\begin{figure}[!t]
\centering
\includegraphics[width=2.9in]{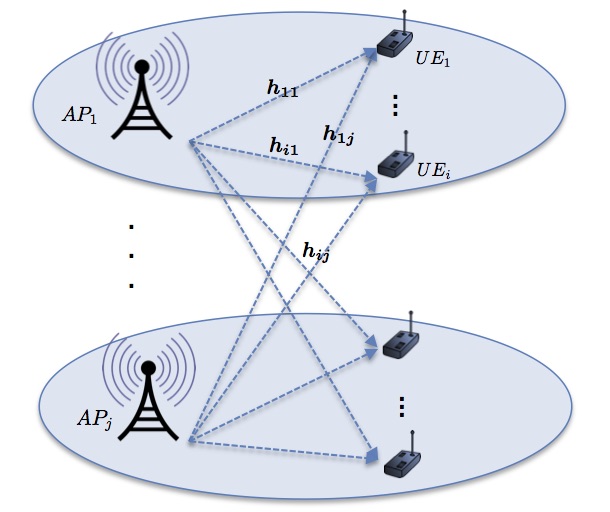}
 \caption{Multi-AP TS MISO SWIPT system}
\label{fig1}
\end{figure}
 According to \eqref{M},  the achievable data rate $R_i$ (bits/sec/Hz) for the $i$th UE can be found from the following equation:
 \begin{equation} \label{R}
  \begin{aligned}
R_i=\log_2(1+\frac{\sum_{j=1}^{N_{AP}}{\text{trace}(\boldsymbol{H_{ij}}\boldsymbol{X_{ij}})}}{\sigma_i^2+\sum_{j=1}^{N_{AP}}\sum_{l=1,l\neq{i}}^{N_{UE}}{{\text{trace}(\boldsymbol{H_{ij}}\boldsymbol{X_{lj}})}}}),
  \end{aligned}
\end{equation}
where $\boldsymbol{X_{ij}}=\boldsymbol{x_{ij}}\boldsymbol{x_{ij}^H}$,
 $\boldsymbol{H_{ij}}=\boldsymbol{h_{ij}}\boldsymbol{h_{ij}^H}$ and therefore $\boldsymbol{X_{ij}},\boldsymbol{H_{ij}}\in\mathbb{C}^{N_{A_j}\times{N_{A_j}}}$ are rank-one matrices for $i\in{\mathcal{N}_{UE}}, j\in{\mathcal{N}_{AP}}$. 
Also the received energy $E_i$ (assuming normalized energy unit of Joule/sec) in the $i$th UE is given by:
 \begin{equation} \label{E}
  \begin{aligned}
E_i=\sum_{j=1}^{N_{AP}}\sum_{l=1}^{N_{UE}}\text{trace}(\boldsymbol{H_{ij}}\boldsymbol{X_{lj}}),
  \end{aligned}
\end{equation}
in which the antenna noise power is neglected.

 The UEs are assumed to be capable of information decoding and energy harvesting using TS receiver architecture. In TS scheme, each reception time frame is divided into two orthogonal time slots, one for ID and the other for EH.  Therefore we have the following equations for the average data rate and the harvested energy at the $i$th UE:
\begin{equation} \label{R1}
  \begin{aligned}
R_i^{TS}(\boldsymbol{X},\alpha_i)={\alpha_i}R_i(\boldsymbol{X}),
  \end{aligned}
\end{equation}
\begin{equation} \label{E1}
  \begin{aligned}
E_{h_i}^{TS}(\boldsymbol{X},\alpha_i)=({1-\alpha_i})\eta_iE_i(\boldsymbol{X}),
  \end{aligned}
\end{equation}
where  $R_i(\boldsymbol{X})$ and $E_i(\boldsymbol{X})$ can be found from \eqref{R} and  \eqref{E}, respectively, $\alpha_i$ is the fraction of time devoted to ID in the $i$th UE and $\eta_i$ denotes the energy harvesting efficiency factor of the $i$th UE.

Our goal is to find the optimal transmit strategies $\boldsymbol{X}=[\boldsymbol{X_{lj}}]_{l\in{\mathcal{N}_{UE}}, j\in{\mathcal{N}_{AP}}}$
 and time switching ratios  $\boldsymbol{\alpha}=[\alpha_i]_{i\in{\mathcal{N}_{UE}}}$, to maximize the performance of all users simultaneously. Since the information data rate and harvested energy are both desirable for each user, we define the utility vector of the $i$th UE by $\boldsymbol{u}_i(\boldsymbol{X},\alpha_i)=[R_i^{TS}(\boldsymbol{X},\alpha_i), E_{h_i}^{TS}(\boldsymbol{X},\alpha_i)]$. Our optimization objective is to maximize the utility vector of the system defined by $\boldsymbol{u}(\boldsymbol{X},\boldsymbol{\alpha})=[\boldsymbol{u}_1(\boldsymbol{X},\alpha_1), \boldsymbol{u}_2(\boldsymbol{X},\alpha_2), ...,
 \boldsymbol{u}_{N_{UE}}(\boldsymbol{X},\alpha_{N_{UE}})]$ jointly via the multi-objective problem formulation. This problem can be written as:
 \begin{equation} \label{P1} 
  \begin{aligned}
    &\underset{\boldsymbol{X}, \boldsymbol{\alpha}}{\text{Maximize}}
& &\boldsymbol{u}(\boldsymbol{X},\boldsymbol{\alpha})\\
 & \text{subject to} 
 & & \mbox{(1)}  \    \sum_{j=1}^{N_{AP}}\sum_{l=1}^{N_{UE}}{\text{trace}({\boldsymbol{X_{lj}}})}\leq{P_{max}}\\
 &
 & & \mbox{(2)}  \   \boldsymbol{X_{lj}}\succeq{0}, \text{Rank}(\boldsymbol{X_{lj}})=1,  \ \forall{l,j}\\
 &
& & \mbox{(3)}  \     \alpha_i\in[0,1], \ \forall{i},
  \end{aligned}
\end{equation}
where the first constraint denotes the average power constraint for APs across all transmitting antennas.
 
The design problem for the ideal SWIPT in which energy is assumed to be extracted simultaneously while information decoding is the same as problem \eqref{P1} but with utility vectors of $\boldsymbol{u}_i(\boldsymbol{X})=[R_i(\boldsymbol{X}), \eta_iE_i(\boldsymbol{X})]$, where  $R_i(\boldsymbol{X})$ and $E_i(\boldsymbol{X})$ can be found from \eqref{R} and  \eqref{E}, respectively. As mentioned earlier, this ideal receiver is not feasible in practice, however,  for theoretical benchmarking, its performance can be used as an upper bound for the performance of TS SWIPT.
\section{Joint transmit precoding and receiver design}
In this section, we study problem \eqref{P1} and propose an algorithm to find the Pareto optimal transmit precoding matrices $\boldsymbol{X}$
 and TS ratios  $\boldsymbol{\alpha}$. We also propose an algorithm for solving the problem in case of fixed switching rates in UEs. 

\subsection{TS SWIPT with adaptive switching rates}
To solve the problem, we first relax the rank constraint and later we show that the optimal solution of the relaxed problem satisfies $\text{Rank}(\boldsymbol{X_{lj}})=1, \forall{l}\in{\mathcal{N}_{UE}}, \forall{j}\in{\mathcal{N}_{AP}}$. To solve the relaxed MOO problem, we use the weighted Chebyshev method which provides complete Pareto optimal set by varying predefined preference parameters $v_i^{(1)},v_i^{(2)}, \ \forall{i}\in{\mathcal{N}_{UE}}$.  This scalarization is equivalent to the following problem:
\\
\begin{equation} \label{P2}
  \begin{aligned}
    &\underset{\boldsymbol{X},\boldsymbol{\alpha},\lambda}{\text{Maximize}}
& & \lambda\\
 & \text{subject to} 
  & & \mbox{(1)}  \    {\alpha_i}R_i(\boldsymbol{X})\geq{\lambda{v_i^{(1)}}}, \ \forall{i}  \\
&
  & & \mbox{(2)}  \     {(1-\alpha_i)}\eta_iE_i(\boldsymbol{X})\geq{\lambda{v_i^{(2)}}}, \ \forall{i}\\
&
 & & \mbox{(3)}  \    \sum_{j=1}^{N_{AP}}\sum_{l=1}^{N_{UE}}{\text{trace}({\boldsymbol{X_{lj}}})}\leq{P_{max}}\\
&
& & \mbox{(4)}  \   \boldsymbol{X_{lj}}\succeq{0}, \ \forall{l,j}\\
&
 & & \mbox{(5)}  \   \alpha_i\in[0,1],  \ \forall{i}.
  \end{aligned}
\end{equation}
This problem is a non-convex SDP due to not only the coupled TS ratios and $R_i$, $E_i$ in the first and second constraints but also the definition of $R_i(\boldsymbol{X})$ as presented in \eqref{R}. Introducing the new variables $R_i,E_i,I_i$ and $\beta_i$, problem \eqref{P2} can be represented as:
\begin{equation} \label{P3} 
  \begin{aligned}
    &\underset{\boldsymbol{X},\alpha_i, \beta_i,R_i,E_i\atop{I_i ,\lambda, \  \forall{i}}}{\text{Maximize}}
& & {\lambda}\\
 & \text{\ \ \ \ \ subject to} 
  & & \mbox{(C1)}  \    {\alpha_i}R_i\geq{\lambda{v_i^{(1)}}}, \ \forall{i}  \\
&
  & & \mbox{(C2)}  \     {\beta_i}\eta_iE_i\geq{\lambda{v_i^{(2)}}}, \ \forall{i}\\
&
  & & \mbox{(C3)}  \     E_i=\sum_{j=1}^{N_{AP}}\sum_{l=1}^{N_{UE}}\text{trace}(\boldsymbol{H_{ij}}\boldsymbol{X_{lj}}),  \ \forall{i}\\
  &
  & & \mbox{(C4)}  \     I_i=\sum_{j=1}^{N_{AP}}\sum_{l=1,l\neq{i}}^{N_{UE}}\text{trace}(\boldsymbol{H_{ij}}\boldsymbol{X_{lj}}),  \ \forall{i}\\
&
 & & \mbox{(C5)}  \     R_i={\log(E_i+{\sigma_i^2})-\log(I_i+{\sigma_i^2})}, \ \forall{i}\\
&
 & & \mbox{(C6)}  \    \sum_{j=1}^{N_{AP}}\sum_{l=1}^{N_{UE}}{\text{trace}({\boldsymbol{X_{lj}}})}\leq{P_{max}} \\
&
& & \mbox{(C7)}  \   \boldsymbol{X_{lj}}\succeq{0}, \ \forall{l,j}\\
&
 & & \mbox{(C8)}  \   \alpha_i+\beta_i=1, \ \forall{i}\\
 &
 & & \mbox{(C9)}  \   \alpha_i\in[0,1], 
   \end{aligned}
\end{equation}
where (C5) is directly obtained from substituting the definition of $E_i$ and $I_i$ in the definition of $R_i$ given by equation \eqref{R}.  
\begin{lem} \label{lem1}
 The constraint (C5) in problem \eqref{P3} can be relaxed to $(\overline{\mbox{C5}})$ defined below:
\begin{equation} \label{}
  \begin{aligned}
(\overline{\mbox{C5}})  \   R_i\leq{\log(E_i+{\sigma_i^2})-\log(I_i+{\sigma_i^2})}.
 \end{aligned}
\end{equation}
\end{lem}
\begin{proof}
See Appendix. A for the proof.
\end{proof}
Considering the result of lemma \ref{lem1}, we study the rank of the optimal precoding matrix by the following proposition: 
\begin{prop} \label{prop1}
The optimal precoding matrices of problem \eqref{P3} are rank-one matrices.
\end{prop}
\begin{proof}
See Appendix B for the proof.
\end{proof}
As mentioned earlier, problem \eqref{P3} is a non-convex SDP. We define $\hat{\lambda}=\log(\lambda)$,
and use the monotonicity and concavity properties of the logarithm function to reformulate the problem as below:
\begin{equation} \label{PP} \tag{$P$}
  \begin{aligned}
    &\underset{\boldsymbol{X},\alpha_i, \beta_i,R_i,E_i\atop{I_i ,\hat{\lambda}, \  \forall{i}}}{\text{Maximize}}& & \hat{\lambda}\\
 & \text{\ \ \ \ \ subject to} 
  & & {(\overline{\mbox{C1}})}  \    \log({\alpha_i})+\log(R_i)\geq{\hat{\lambda}+\log(v_i^{(1)})}  \\
&
  & & {(\overline{\mbox{C2}})}  \     {\log(\beta_i)}+\log(\eta_iE_i)\geq{\hat{\lambda}+\log({v_i^{(2)}})}\\
  &
  & & \mbox{(C3)-(C4)}\\
    &
 & & (\overline{\mbox{C5}})  \   R_i\leq{\log(E_i+{\sigma_i^2})-\log(I_i+{\sigma_i^2})}\\
 &
 & & \mbox{(C6)-(C9)}. \\
    \end{aligned}
\end{equation}
Now the nonconvexity of problem \eqref{P3} is concentrated in inequality $(\overline{\mbox{C5}})$.
However, problem \eqref{PP} can be considered as a DC (difference of convex)
programming since $(\overline{\mbox{C5}})$ is the difference of two convex functions ($R_i-\log(E_i+{\sigma_i^2})$, $-\log(I_i+{\sigma_i^2})$). Therefore, it can be solved using local optimization method of convex-concave procedure (CCP) \cite{ccp}. CCP is a majorization-minimization \cite{MM} algorithm that solves DC programs as a sequence of convex programs by linearizing the concave part $\log(I_i)$ around the current iteration solution of $I_i$. 
To this end, we use the first order Taylor expansion and replace problem \eqref{PP} in the $k$th step by the following subproblem:
\begin{equation} \label{Pk} \tag{$P_k$}
  \begin{aligned}
    &\underset{\boldsymbol{X},\alpha_i, \beta_i,R_i,E_i\atop{I_i ,\hat{\lambda}, \  \forall{i}}}{\text{Maximize}}& & \hat{\lambda}\\
 & \text{\ \ \ \ \ subject to} 
  & & {(\overline{\mbox{C1}}),(\overline{\mbox{C2}}),\mbox{(C3)-(C4)}}\\
    &
 & & (\overline{\mbox{C5}})  \   R_i\leq\log(E_i+{\sigma_i^2})-\\
 &
 && (\log(I_i^k+{\sigma_i^2})+\frac{1}{{I_i^k}+{\sigma_i^2}}(I_i-{I_i^k}))\\
 &
 & & \mbox{(C6)-(C9)}. \\
    \end{aligned}
\end{equation}
This problem is a convex SDP and it can be solved by standard optimization techniques such as Interior-point Method. In this paper, we have used the CVX package to solve \eqref{Pk}.
The linearization point is updated with each iteration until it satisfies the termination criterion as described in Algorithm \ref{Alg1}. 
\begin{algorithm} [!t]
    \caption{CCP Algorithm for TS SWIPT with adaptive switching rates}
    \begin{spacing} {1} 
    \label{Alg1}
\begin{algorithmic} [1]
  \STATE \textbf{Step 0}:  Choose an initial point $I_i^0$  inside the convex set defined by (C1)-(C4), (C6)-(C9), $\gamma\in{\mathbb{R}}$ and a given tolerance $\epsilon>0$. Set $k:=0$.\\ 
\STATE  \textbf{Step 1}: For a given ${I}_i^k$, solve the convex SDP of \eqref{Pk} to obtain the solution $\hat{I}_i(I^k_i)$.\\ 
\STATE  \textbf{Step 2}:  If $\|\hat{I}_i(I^k_i)-{I}^k_i\|\leqslant{\epsilon}$  then stop. Otherwise  set ${I}^k_i=I^k_i+\gamma(\hat{{I}_i}({I}^k_i)-{I}^k_i)$.
\STATE  \textbf{Step 3}: increase $k$ by 1 and go back to step 1.
\end{algorithmic}
\end{spacing}
\end{algorithm}
We study the convergence of algorithm \ref{Alg1} by the following theorem.
\begin{prop} \label{prop2}
Algorithm \ref{Alg1} converges to a stationary point of problem \eqref{PP}.
\end{prop}
\begin{proof}
See Appendix C for the proof.
\end{proof}
\subsection{TS SWIPT with fixed switching rates}
In this scenario, we denote the fixed ID and EH switching rates for the $i$th UE by $\alpha_i$ and $\beta_i$  ($\alpha_i+\beta_i=1$). 
It is clear that the ideal SWIPT, mentioned earlier, is a special case of this scenario with $\alpha_i=\beta_i=1,  \forall{i}$. Using the same strategy as for TS SWIPT with adaptive switching rates, the relaxed SOO problem for this case is given below:
\begin{equation} \label{Q1} 
  \begin{aligned}
    &\underset{\boldsymbol{X},R_i,E_i\atop{I_i ,\lambda, \  \forall{i}}}{\text{Maximize}}
& & {\lambda}\\
 & \text{\ \ \ \ \ subject to} 
  & & \mbox{(C1)}  \    \alpha_iR_i\geq{\lambda{v_i^{(1)}}}, \ \forall{i}  \\
&
  & & \mbox{(C2)}  \     \eta_i\beta_iE_i\geq{\lambda{v_i^{(2)}}}, \ \forall{i}\\
&
  & & \mbox{(C3)}  \     E_i=\sum_{j=1}^{N_{AP}}\sum_{l=1}^{N_{UE}}\text{trace}(\boldsymbol{H_{ij}}\boldsymbol{X_{lj}}),  \ \forall{i}\\
  &
  & & \mbox{(C4)}  \     I_i=\sum_{j=1}^{N_{AP}}\sum_{l=1,l\neq{i}}^{N_{UE}}\text{trace}(\boldsymbol{H_{ij}}\boldsymbol{X_{lj}}),  \ \forall{i}\\
&
 & & \mbox{(C5)}  \     R_i={\log(E_i+{\sigma_i^2})-\log(I_i+{\sigma_i^2})}, \ \forall{i}\\
&
 & & \mbox{(C6)}  \    \sum_{j=1}^{N_{AP}}\sum_{l=1}^{N_{UE}}{\text{trace}({\boldsymbol{X_{lj}}})}\leq{P_{max}} \\
&
& & \mbox{(C7)}  \   \boldsymbol{X_{lj}}\succeq{0}, \ \forall{l,j}.
   \end{aligned}
\end{equation}
This problem is also a non-convex SDP because of the nonlinear equality in (C5). Before solving the problem, we study the rank of the optimal solution in the following proposition. 
\begin{prop} \label{prop3}
The optimal precoding matrices of problem \eqref{Q1} are rank-one matrices.
\end{prop}
\begin{proof}
See Appendix D for the proof.
\end{proof}
As mentioned, the main difficulty of problem \eqref{Q1} is concentrated in the nonlinear equality of  (C5). It can be inferred from proof of proposition \ref{prop3} that this equality constraint can not be relaxed to inequality constraint in this problem. This issue can be overcome by solving the problem iteratively and  linearizing (C5) around the current iteration point while maintaining the remaining convexity of the original problem. This method is called sequential convex programming (SCP) \cite{scp} which is an iterative local optimization method that generates a sequence of solutions to the convex subproblems. We use the first order Taylor expansion to write the linearized version of (C5) as follows:
 \begin{equation} \label{}
  \begin{aligned}
  (\overline{\mbox{C5}})  \   R_i&\simeq{R_i(E_i^0,I_i^0)+\nabla^T{R_i}(E_i^0,I_i^0)[E_i-E_i^0,I_i-I_i^0]}\\
  &=\log(\frac{\sigma_i^2+E_i^0}{\sigma_i^2+{I_i^0}})\\
&+\frac{1}{\sigma_i^2+{E_i^0}}(E_i-{E_i^0})-\frac{1}{\sigma_i^2+{I_i^0}}(I_i-{I_i^0}),
  \end{aligned}
\end{equation}
where $E_i^0$ and $I_i^0$ are the points around which the equation is linearized. 
Now we can replace problem \eqref{Q1} in the $k$th step by the following subproblem:
  \begin{equation} \label{Qk} \tag{$Q_k$}
  \begin{aligned}
   &\underset{\boldsymbol{X},R_i,E_i\atop{I_i ,\lambda, \  \forall{i}}}{\text{Maximize}}& & {\lambda}\\
 & \text{\ \ \ \ \ subject to} 
  & & \mbox{(C1)-(C4)}\\
    &
 & & (\overline{\mbox{C5}})  \   R_i\simeq{\log(\frac{\sigma_i^2+E_i^k}{\sigma_i^2+{I_i^k}})}\\
 &
 &&+\frac{1}{\sigma_i^2+{E_i^k}}(E_i-{E_i^k})-\frac{1}{\sigma_i^2+{I_i^k}}(I_i-{I_i^k}),\\
 &
 & & \mbox{(C6)-(C7)} \\
   \end{aligned}
\end{equation}
This problem is a convex SDP and it can be solved by standard optimization techniques similarly to subproblem \eqref{Pk}. The linearization point is updated with each iteration until it satisfies the termination criterion as described in Algorithm 2. 
\begin{algorithm} [!t]
    \caption{SCP Algorithm for TS SWIPT with fixed switching rates}
    \begin{spacing} {1} 
        \label{Alg2}
\begin{algorithmic} [1]
  \STATE \textbf{Step 0}:  Choose an initial point $\boldsymbol{w}_i^0=[E_i^0, I_i^0]$  inside the convex set defined by (C1)-(C4), (C6)-(C7), $\gamma\in{\mathbb{R}}$ and a given tolerance $\epsilon>0$. Set $k:=0$.\\ 
\STATE  \textbf{Step 1}: For a given $\boldsymbol{w}_i^k$, solve the convex SDP of \eqref{Qk} to obtain the solution $\hat{\boldsymbol{w}}_i(\boldsymbol{w}^k_i)=[\hat{E}_i(E^k_i), \hat{I}_i(I^k_i)]$.\\ 
\STATE  \textbf{Step 2}:  If $\|\hat{\boldsymbol{w}}_i(\boldsymbol{w}^k_i)-\boldsymbol{w}^k_i\|\leqslant{\epsilon}$  then stop. Otherwise  set $\boldsymbol{w}^k_i=\boldsymbol{w}^k_i+\gamma(\hat{\boldsymbol{w}_i}(\boldsymbol{w}^k_i)-\boldsymbol{w}^k_i)$.
\STATE  \textbf{Step 3}: increase $k$ by 1 and go back to step 1.
\end{algorithmic}
\end{spacing}
\end{algorithm}

The local convergence of Algorithm \ref{Alg2} to a stationary point of problem \eqref{Q1} is proven in \cite{dinh}, under mild assumptions, and the rate of
convergence is shown to be linear.
\section{Numerical Results}
In this section, we present numerical results to demonstrate the performance of the proposed multi-objective precoding and TS design algorithm in MISO SWIPT systems. We consider \textcolor{black}{the network setup} shown in Fig. \ref{fig2}, consisting $N_{AP}=2$ APs equipped with $N_{A_1}=N_{A_2}=2$ antennas  and $N_{UE}$ single antenna TS SWIPT sensors. 
Transmission channel gains, $\boldsymbol{h}_{ij}, \forall{i}\in{\mathcal{N}_{UE}}, j=1,2$, depend on the location of sensors with respect to APs and the channel fading model. Sensors in each \textcolor{black}{UEs are assumed to be distributed uniformly in the geographical area between two circles with radius $d_{min}=2$ m and $d_{max}=10$ m, respectively}. Distance of APs from each other is denoted by $D$ as shown in Fig. \ref{fig2} and is set to $D=20$ m unless it is stated clearly with a different value. The line-of-sight (LoS) component is dominant in these short distances, thus at each location channel gains are generated with Rician fading. The Rician factor, defined as the ratio of signal power in dominant component over the scattered power, is set to $K=3.5$ dB  and path-loss exponent of $3$ is considered.  Noise powers are assumed to be $\sigma_i^2=-90$ dBm $\forall{i}\in{\mathcal{N}_{UE}}$  and the maximum total power budget is set to $P_{max}=1$ watt.
\begin{figure}[!t]
\centering
\includegraphics[width=3in]{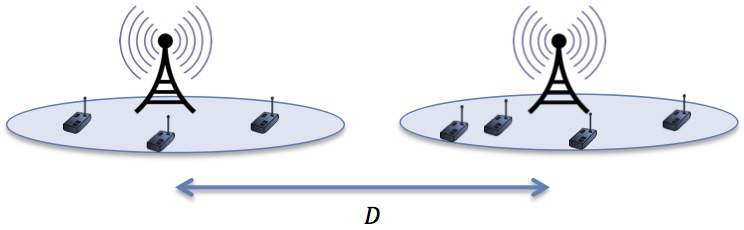}
 \caption{Simulation scheme}
\label{fig2}
\end{figure}
To illustrate the Pareto boundary for TS and ideal SWIPT systems, we solve the optimization problems using Algorithm \ref{Alg1} and \ref{Alg2} in several directions by changing the preference weights $v_i^{(1)},v_i^{(2)}, \forall{i}$. 
It should be mentioned that the case of having a number of data only receivers or energy only harvesters are particular cases of this setup with changing the preference parameters only and is not considered in these simulations.
\begin{figure}[t]%
\centering
\includegraphics[width=4.5in]{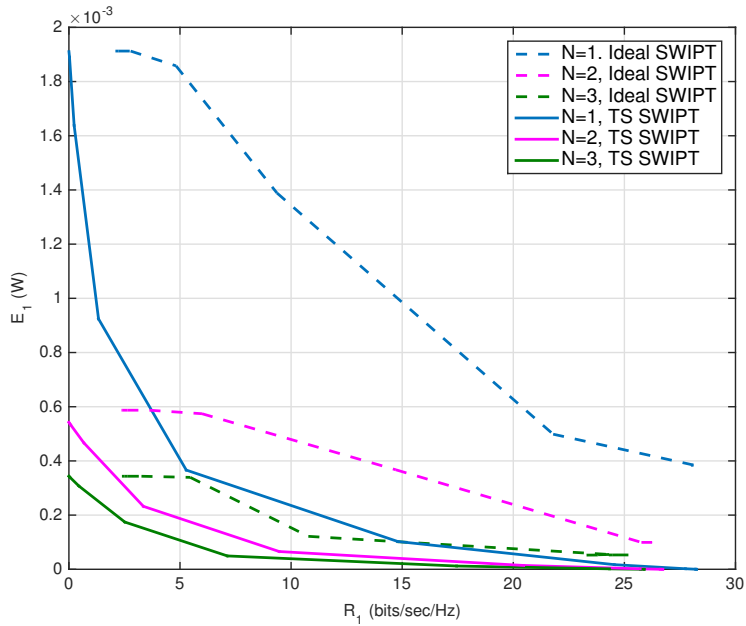}%
\caption{Pareto frontier of TS SWIPT and ideal SWIPT}
\label{fig3}
\end{figure}%

First we investigate the trade-off between harvested energy and information data rate in a symmetric setup which includes $N=\frac{N_{UE}}{2}=1,2,3$ UEs in each area with preference weights of $v_i^{(1)}=\theta_1, v_i^{(2)}=1 \ \forall{i}$. We search for the optimal solutions in  different directions by changing the value of $\theta_1$. Fig.\ref{fig3} shows the Pareto frontier of  the first TS SWIPT UE  for different number of UEs. 

It can be observed that the average harvested energy, $E_{h_1}^{TS}$, is a monotonically decreasing function of the achievable data rate $R_1^{TS}$. This result shows that these two objectives are generally conflicting and any resource allocation algorithm that maximizes the harvested energy cannot maximize the data rate. 
Also it can be seen that, increasing the number of UEs highly affects the possible harvestable amount of energy at each UE. This result is expectable, due to the fixed total power consumption assumption and the direct impact of transmit power on the received energy. \\
Pareto frontier of the infeasible ideal SWIPT, in which EH and ID are performed simultaneously is also shown in this figure as an upper bound.
It can be seen that for each $N$, the maximum harvested energy and 
the maximum achieved data rate of ideal SWIPT is equal to the maximum harvested energy and the maximum achieved data rate of TS SWIPT. However, ideal SWIPT can still have a nonzero minimum data rate and minimum harvested energy in these two extreme cases, because of the assumption of simultaneous EH and ID. Moreover, comparing the results of the ideal and TS SWIPT in different number of UEs reveals that the energy harvesting performance loss of TS with respect to the ideal SWIPT increases with decreasing the number of UEs.

The effect of multi-user interference on the trade-off of data rate and harvested energy is shown in Fig. \ref{fig6}. In Fig. \ref{fig4}, we have plotted the Pareto frontiers of the first UE for different AP distances of $D=5,10,15,20$ m. As can be seen, the maximum possible harvested energy increases significantly by decreasing the AP distances to $D=5,10$ m. However, the maximum data rate changes very slightly with decreasing $D$. To have a more detailed comparison, we have plotted the curves in logarithmic scale in Fig. \ref{fig5}. According to this figure, in low data rates, the system can benefit from the interference for energy harvesting by decreasing the distance between APs or equivalently by densifying the network more. 
\begin{figure}[H]%
\centering
\begin{subfigure}{1\columnwidth}
\centering
\includegraphics[width=4.5in]{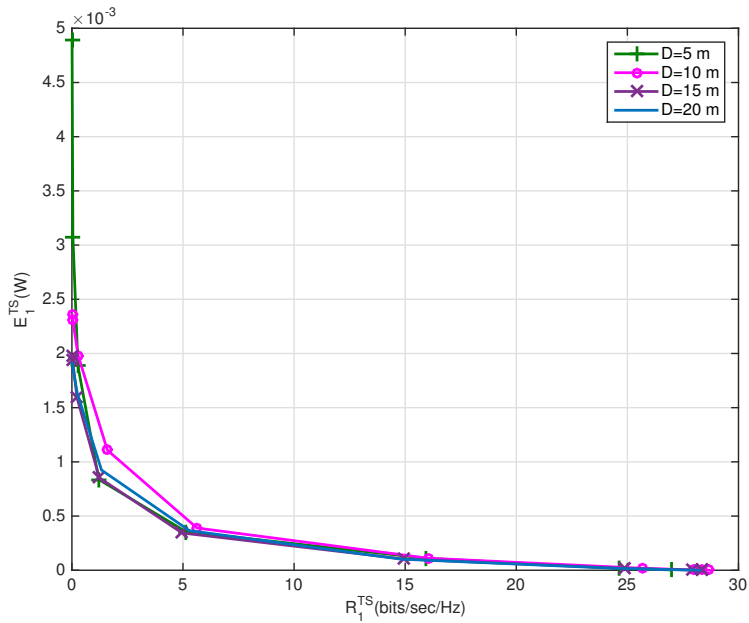}%
\caption{}%
\label{fig4}%
\end{subfigure}\hfill%
\begin{subfigure}{1\columnwidth}
\centering
\includegraphics[width=4.5in]{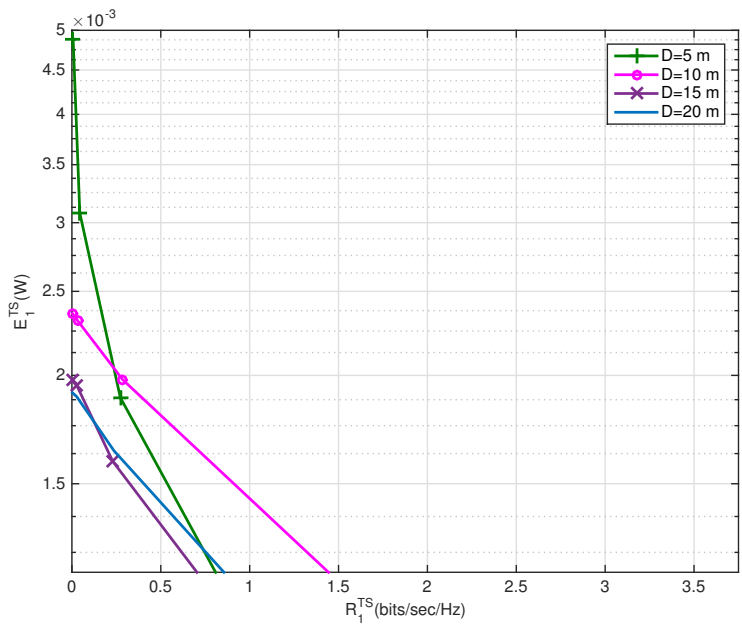}%
\caption{}%
\label{fig5}%
\end{subfigure}%
\caption{Effect of AP distance $D$ on Pareto frontier of the first TS SWIPT UEs}
\label{fig6}
\end{figure}
To study the effect of interference on the energy harvesting, we have also plotted in Fig. \ref{fig7} the ratio $r$, of the maximum amount of energy which is harvested in the first UE  from its nearest AP to the total amount of maximum harvested energy from both APs versus $D$ for two Rician parameters of $K=0$ (Rayleigh channel) and $K=3.5$ dB and $d_{max}=5,10,15$ m. As can be seen, by increasing the distance of APs in all the cases, $r$ goes to one which means that the UE harvests energy directly from its nearest AP. However, in closer AP distances, a percentage ($1-r$) of total amount of harvested energy is harvested from the other AP. As shown in Fig. \ref{fig7}, this percentage increases with increasing the maximum possible distance of UEs from APs, i.e. $d_{max}$. Moreover, the results show that with $d_{max}=5$, the performance will be nearly the same for both Rayleigh and Rician fading channels. However, in larger cells more energy will be harvested from the nearest AP in Rician channel due to the existence of LoS. 
\begin{figure}[h]%
\centering
\includegraphics[width=4.5in]{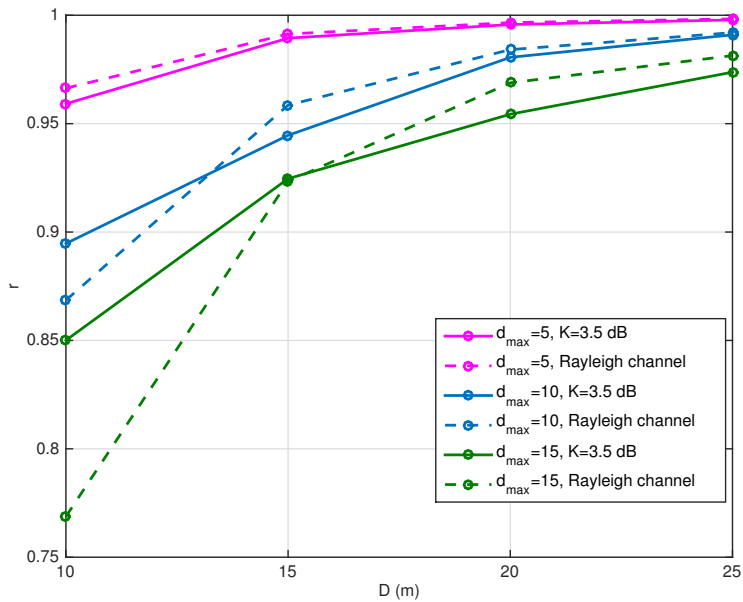}%
\caption{The ratio  of the maximum amount of energy which is harvested in the first UE from its nearest AP to the total amount of maximum harvested energy from both APs versus $D$}
\label{fig7}
\end{figure}%

To study the trade-off between different users, we choose different preference weights for $N_{UE}=2$ UEs by setting $v_1^{(1)}=\theta_1 \theta_2 , v_1^{(2)}=\theta_2$ and $v_2^{(1)}=\theta_1, v_2^{(2)}=1$. Fig. \ref{fig10} shows the Pareto frontiers of these two UEs for $\theta_2=1,5,10,15$.  As illustrated, for $\theta_2=1$ both UEs have the same Pareto frontiers. To benefit from better performance in the first UE, we increase the $\theta_2$. It can be inferred from Fig. \ref{fig10} that this superior performance is not achievable by only adapting the TS ratio. Consequently beamformers will be aligned toward the first UE by allocating more power to the first AP which results in increasing the $E_{h_1}^{TS}$ without increasing the interference on the first UE. Hence the maximum data rate and harvested energy both decrease in the second UE. Therefore, it is inferred that improving the performance of one user by increasing its preference weight is at the expense of destroying the performance of the other user drastically.

\begin{figure}[h]%
\centering
\begin{subfigure}{1\columnwidth}
\centering
\includegraphics[width=4.5in]{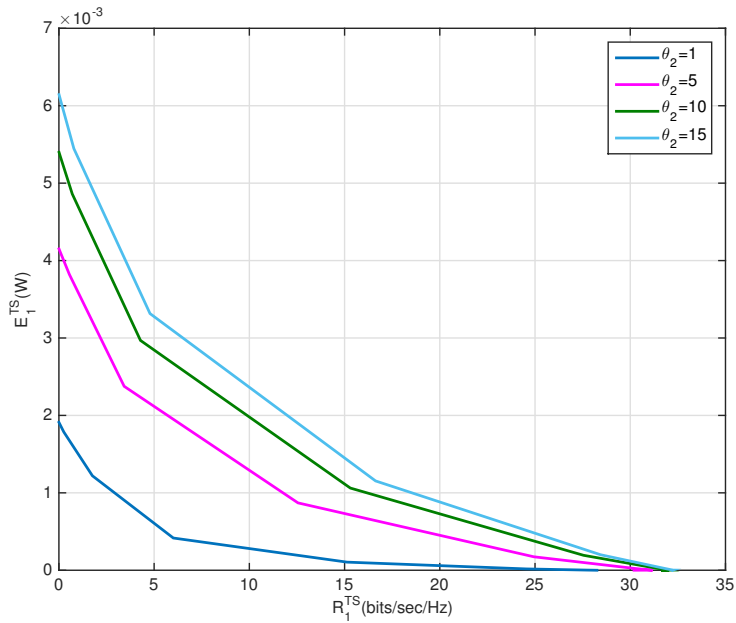}%
\caption{}%
\label{fig8}%
\end{subfigure}\hfill%
\begin{subfigure}{1\columnwidth}
\centering
\includegraphics[width=4.5in]{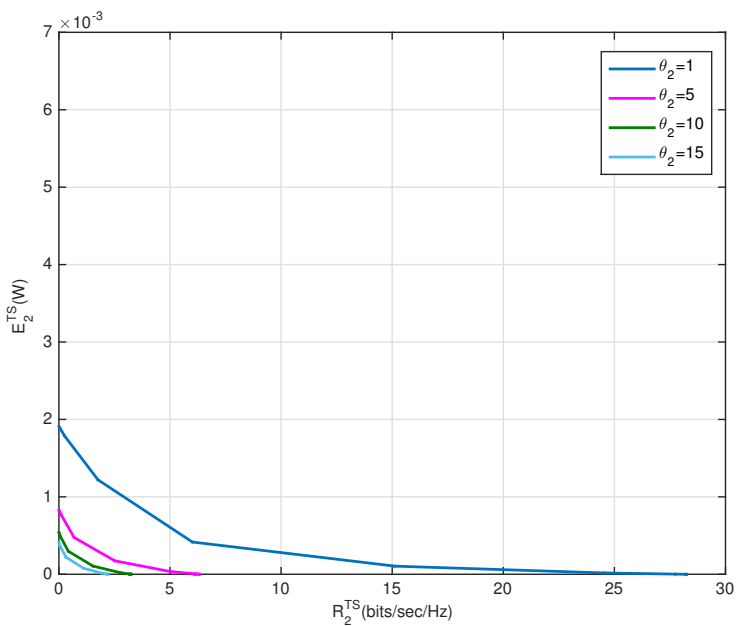}%
\caption{}%
\label{fig9}%
\end{subfigure}%
\caption{Pareto frontier of first (a) and second (b) TS SWIPT UEs with different priorities}
\label{fig10}
\end{figure}
 
In Fig. \ref{fig11}, we show the effect of optimizing the TS ratios on the energy harvested-data rate performance. In this figure, we have compared the Pareto frontier of the adaptive TS SWIPT with the Pareto frontier of TS SWIPT with fixed switching rates of $\alpha_i=0.1,0.3,0.5,0.7,0.9, \forall{i}$ for one realization of the channel.  As can be seen, the lower the $\alpha_1$s, the higher the maximum energy harvested and the lower the maximum data rate, while the optimal TS leverages the best possible harvested energy and data rate by optimizing $\alpha_i$ jointly with the transmit strategy.

 \begin{figure}[!h]%
\centering
\includegraphics[width=4.5in]{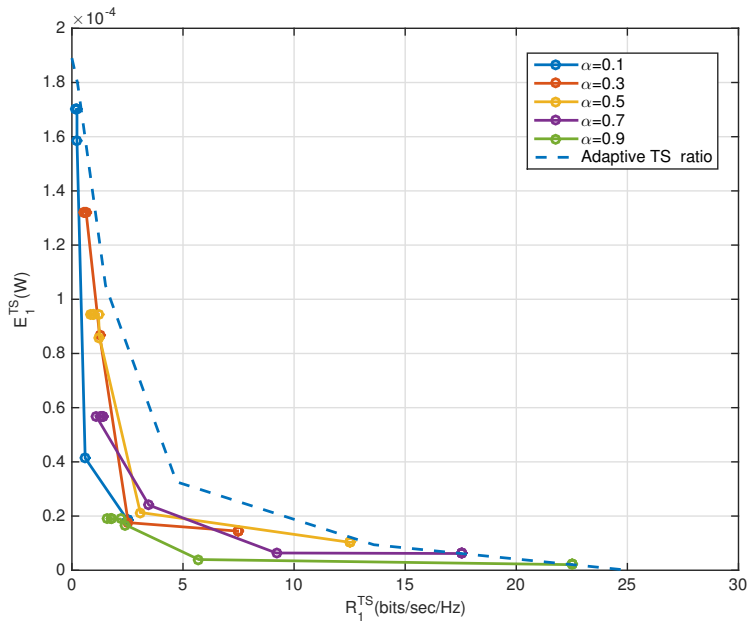}%
\caption{Pareto frontier of TS SWIPT with fixed and adaptive TS ratios}
\label{fig11}
\end{figure}%
\section{Conclusion}
In this paper, we studied the joint transmit precoding and receiver time switching design for downlink MISO SWIPT systems. The design problem was formulated as a non-convex MOO problem with the goal of maximizing the harvested energy and information data rates for all users simultaneously. The proposed MOO problem was scalarized employing the weighted Chebyshev method. This problem is a non-convex SDP which is relaxed and solved using convex-concave procedure which is a majorization-minimization algorithm. The trade-off between the harvested energy and the information data rate was studied by means of numerical results. The numerical results showed that in higher number of UEs the harvested energy loss to gain the desired data rate in TS SWIPT with respect to the ideal SWIPT is lower. 
Also it was shown that cooperation among multiple APs can be used to drastically increase the achievable trade-off of one UE but the effect on the trade-off of other UEs could be detrimental. Moreover, by studying the effect of multi-user interference on the performance of the system, it was inferred that in low-rate devices we can harvest more energy by densifying the network. 

\appendices
\section{Proof of Lemma \ref{lem1}}
The lemma can be proved by analysing the Karush-Kuhn-Tucker (KKT) \cite{boyd} 
conditions for the SDP problem \eqref{P3}. 
Lagrangian for this problem is defined as: \\
 \begin{equation} \label{L}
  \begin{aligned}
&\mathcal{L}(\gamma_i^1,\gamma_i^2,\gamma_i^3,\gamma_i^4,\gamma^5,\gamma^6,\gamma_i^7,\gamma_i^8,\gamma_i^9, \Gamma_{lj})=\\
&\lambda-\sum_{i=1}^{N_R}{\gamma_i^1(\lambda{v_i^{(1)}}-{\alpha_i}R_i)}-\\
&\sum_{i=1}^{N_{UE}}{\gamma_i^2(\lambda{v_i^{(2)}}-{\eta_i}{\beta_i}E_i)}-\\
&\sum_{i=1}^{N_{UE}}{\gamma_i^3(R_i-\log(E_i+{\sigma_i^2})+\log(I_i+{\sigma_i^2}))}-\\
&\sum_{i=1}^{N_{UE}}{\gamma_i^4(E_i-\sum_{j=1}^{N_{AP}}\sum_{l=1}^{N_{UE}}\text{trace}(\boldsymbol{H_{ij}}\boldsymbol{X_{lj}})})-\\
&\sum_{i=1}^{N_{UE}}{\gamma_i^5(I_i-\sum_{j=1}^{N_{AP}}\sum_{l=1,l\neq{i}}^{N_{UE}}\text{trace}(\boldsymbol{H_{ij}}\boldsymbol{X_{lj}}))}-\\
&{\gamma^6(\sum_{j=1}^{N_{AP}}\sum_{l=1}^{N_{UE}}{\text{trace}({\boldsymbol{X_{lj}}})}-{P_{max}}})-\\
&\sum_{l=1}^{N_{UE}}{\gamma_i^7(1- \alpha_i-\beta_i)}+\sum_{l=1}^{N_{UE}}{\gamma_i^8\alpha_i+\gamma_i^9\beta_i}+\\
&\sum_{j=1}^{N_{AP}}\sum_{l=1}^{N_{UE}}\text{trace}(\boldsymbol{\Gamma_{lj}}\boldsymbol{X_{lj}}),
  \end{aligned}
\end{equation}
where $\gamma_i^1,\gamma_i^2,\gamma^6,\gamma_i^8,\gamma_i^9\geq{0}$ and $\boldsymbol{\Gamma_{lj}}\succeq{0}$ are the Lagrangian multipliers. 
First we show that constraint (C6) always meet its boundary. 
Assume that the optimal  $\sum_{j=1}^{N_{AP}}\sum_{l=1}^{N_{UE}}{\text{trace}({\boldsymbol{X_{lj}}})}<{P_{max}}$, then we can multiply all optimal precoding matrices  $\boldsymbol{X_{lj}}$ by a factor of $\alpha>1$, for which according to \eqref{R} and \eqref{E} both $R_i$ and $E_i$ and therefore ${\lambda}$ would increase. Thus   $\sum_{j=1}^{N_{AP}}\sum_{l=1}^{N_{UE}}{\text{trace}({\boldsymbol{X_{lj}}})}={P_{max}}$ and as a result $\gamma^6>0$.
The useful KKT conditions for this problem can be written as:
  \begin{align}
  \label{A6}
&\mathop{\nabla}\limits_{\lambda}{\mathcal{L}}=1-\sum_{i=1}^{N_{UE}}(\gamma_i^1{v_i^{(1)}}+\gamma_i^2{v_i^{(2)}})=0\\ 
\label{A1}
&\mathop{\nabla}\limits_{R_i}{\mathcal{L}}={\gamma_i^1{\alpha_i}}-\gamma_i^3=0\\ 
\label{A9}
&\mathop{\nabla}\limits_{E_i}{\mathcal{L}}={\gamma_i^2\eta_i{\beta_i}}+\frac{\gamma_i^3}{E_i+{\sigma_i^2}}-\gamma_i^4=0\\
\label{A10}
&\mathop{\nabla}\limits_{I_i}{\mathcal{L}}=-\frac{\gamma_i^3}{I_i+{\sigma_i^2}}-\gamma_i^5=0\\
\label{A4}
&\mathop{\nabla}\limits_{\alpha_i}{\mathcal{L}}={\gamma_i^1}{R_i}+\gamma_i^7+\gamma_i^8=0\\
\label{A5}
&\mathop{\nabla}\limits_{\beta_i}{\mathcal{L}}={\gamma_i^2}{\eta_i}{E_i}+\gamma_i^7+\gamma_i^9=0\\
\label{A11}
&\mathop{\nabla}\limits_{\boldsymbol{X_{lj}}}{\mathcal{L}}={-\gamma^6}\boldsymbol{I_{N_{A_j}}}+\sum_{i=1}^{N_{UE}}{(\gamma_i^4+\gamma_i^5)\boldsymbol{H_{ij}}-{\gamma_l^5}\boldsymbol{H_{lj}}}+\boldsymbol{\Gamma_{lj}}=0\\
\label{A2}
&\gamma_i^8{\alpha_i}=0,\\ 
\label{A3}
&{\gamma_i^9}{\beta_i}=0,\\
\label{A12}
&\boldsymbol{\Gamma_{lj}}\boldsymbol{X_{lj}}=0,
 \end{align}

We show that $\gamma_i^1>{0}, \ \forall{i}$ and therefore according to \eqref{A1}, the Lagrangian multiplier related to (C5) is positive, i.e. $\gamma_i^3>{0}, \ \forall{i}$. This verifies that (C5) can be replace by $(\overline{\mbox{C5}})$ in \eqref{P3}. 
\\
Since the optimal  $\alpha_i\neq{0}$, slackness complementary condition \eqref{A2}  result in 
$\gamma_i^8=0$. Thus \eqref{A4} , \eqref{A5} yeild in:
\begin{equation} \label{eq}
  \begin{aligned}
  {\gamma_i^1}{R_i}={\gamma_i^2}{\eta_i}{E_i}+\gamma_i^9, \  \forall{i},
 \end{aligned}
\end{equation}
If $\gamma_i^1={0}$ then since $\gamma_i^2,\gamma_i^9\geq{0}$, for nonzero solution $E_i>0$, we have $\gamma_i^2=\gamma_i^9={0}$. Therefore $\gamma_i^1={0}$ leads to $\gamma_i^2={0}$ which considering \eqref{A6}, states that all ${\gamma_i^1}$ can not be zero at the same time. 
Suppose that there exist an $\hat{i}$ for which $\gamma_{\hat{i}}^1={0}$ and $\gamma_{i}^1>{0}, \forall{i\neq{\hat{i}}}$. Then according to \eqref{A1}, \eqref{A10} and \eqref{eq}, \eqref{A9} $\gamma_{\hat{i}}^2=\gamma_{\hat{i}}^3=\gamma_{\hat{i}}^4=\gamma_{\hat{i}}^5={0}$. 
According to \eqref{A11} the optimal dual matrix $\Gamma_{lj}$ satisfies the following equalities:
\begin{equation}\label{gam1}
  \begin{aligned}
\boldsymbol{\Gamma_{\hat{i}j}}={\gamma^6}\boldsymbol{I_{N_{A_j}}}-\sum_{i=1, i\neq{\hat{i}}}^{N_{UE}}(\gamma_i^4+\gamma_i^5)\boldsymbol{H_{ij}},
  \end{aligned}
\end{equation}
and
\begin{equation}\label{gam2}
  \begin{aligned}
\boldsymbol{\Gamma_{lj}}=&{\gamma^6}\boldsymbol{I_{N_{A_j}}}-\sum_{i=1, i\neq{\hat{i}}}^{N_{UE}}(\gamma_i^4+\gamma_i^5)\boldsymbol{H_{ij}}+{\gamma_l^5}\boldsymbol{H_{lj}}\\=&
\boldsymbol{\Gamma_{\hat{i}j}}+{\gamma_l^5}\boldsymbol{H_{lj}}, \ \ \ \forall{l}\neq{\hat{i}}.
  \end{aligned}
\end{equation}
We know that $\boldsymbol{\Gamma_{\hat{i}j}}$ is  positive definite. If the minimum eigenvalue of $\boldsymbol{\Gamma_{\hat{i}j}}$ is zero, there exists one vector $\boldsymbol{a}\neq{0}$ such that $\boldsymbol{a}^H{\boldsymbol{\Gamma_{\hat{i}j}}}\boldsymbol{a}=0$. Thus using \eqref{gam2} we have:
\begin{equation}
 \begin{aligned}
\boldsymbol{a}^H{\boldsymbol{\Gamma_{lj}}}\boldsymbol{a}={\gamma_l^5}\boldsymbol{a}^H{\boldsymbol{H_{lj}}}\boldsymbol{a}\geq{0}, \ \forall{l}\neq{\hat{i}},
  \end{aligned}
\end{equation}
because $\boldsymbol{\Gamma_{ij}}$ is positive definite. But  \eqref{A10} imply that $\gamma_l^5<{0}, \forall{l\neq{\hat{i}}}$ which results in:
\begin{equation}
 \begin{aligned}
\boldsymbol{h_{lj}^H}\boldsymbol{a}={0}, \forall{l}\neq{\hat{i}}.
  \end{aligned}
\end{equation}
According to \eqref{gam1} this leads to:
\begin{equation}
 \begin{aligned}
\boldsymbol{a}^H{\boldsymbol{\Gamma_{\hat{i}j}}}\boldsymbol{a}={\gamma^6}{\boldsymbol{a}^H}\boldsymbol{a}=0, \ \forall{j},
  \end{aligned}
\end{equation}
which is in contradiction with $\boldsymbol{a}\neq{0}$ and $\gamma^6>{0}$. Thus $\boldsymbol{\Gamma_{\hat{i}j}}$ is strictly positive definite and this together with \eqref{A12} result in $\boldsymbol{X_{\hat{i}j}}=0, \forall{j}$ which is not the optimal solution. Therefore  ${\gamma_{\hat{i}}^1\neq0}$, i.e., ${\gamma_{{i}}^1>0}, \forall{i}$ and the lemma will be proved. 
\section{Proof of proposition \ref{prop1}}
According to \eqref{A11} the optimal dual matrix $\Gamma_{ij}$ satisfies the following equality:
\begin{equation}\label{gam}
  \begin{aligned}
\boldsymbol{\Gamma_{lj}}={\gamma^6}\boldsymbol{I_{N_{A_j}}}-\sum_{i=1}^{N_{UE}}(\gamma_i^4+\gamma_i^5)\boldsymbol{H_{ij}}+{\gamma_l^5}\boldsymbol{H_{lj}}.
  \end{aligned}
\end{equation}
We know that $\boldsymbol{\Gamma}_{lj}\succeq{0}$ and $\boldsymbol{H_{lj}}=\boldsymbol{h_{lj}}\boldsymbol{h_{lj}^H}$
is a rank-one positive semidefinite matrix. Also according to Lemma \eqref{lem1} $\gamma_i^1>{0}, \forall{i}$ and therefore \eqref{A1}, \eqref{A10} imply that $\gamma_i^5<{0}, \forall{i}$. So we have:
\begin{equation}\label{PD}
 \begin{aligned}
\boldsymbol{A_{j}}=&{\gamma^6}{I_{N_{A_j}}}-\sum_{i=1}^{N_{UE}}(\gamma_i^4+\gamma_i^5)\boldsymbol{H_{ij}}\\
=&\boldsymbol{\Gamma_{lj}}-{\gamma_l^5}\boldsymbol{H_{lj}}\succeq{0}, \ \forall{l,j}.
  \end{aligned}
\end{equation}
Using the same method as in proof of lemma 1, it can be shown that $\boldsymbol{A_{j}}$ is strictly positive definite
and $\text{Rank}(\boldsymbol{A_{j}})=N_{A_j}$. As a result it follows from \eqref{gam} that:
\begin{equation}
  \begin{aligned}
\text{Rank}(\boldsymbol{\Gamma_{lj}})&\geq{\text{Rank}(\boldsymbol{A_{j}})-\text{Rank}(-{\gamma_l^5}\boldsymbol{H_{lj}})}\\
& \geq{N_{A_j}-1}, \forall{l,j}.
  \end{aligned}
\end{equation}
On the other side, according to \eqref{A12}  if $\text{Rank}(\boldsymbol{\Gamma_{lj}})=N_{A_j}, \forall{l,j}$  we have $\boldsymbol{X_{lj}}=0, \forall{l,j}$ which is not the optimal solution. Therefore $\text{Rank}(\boldsymbol{\Gamma_{lj}})=N_{A_j}-1, \forall{l,j}$ and rank-null theorem leads to $\text{Rank}(\boldsymbol{X_{lj}})=1, \forall{l,j}$.
\section{Proof of proposition \ref{prop2}}
It can be easily shown that if $I^k$ is the stationary point of subproblem \eqref{Pk}, i.e. fulfilling the KKT conditions of subproblem \eqref{Pk}, it is also a stationary point of problem \eqref{PP}. 
Convergence of constrained CCP  is studied in \cite{convg} using Zangwill's global convergence theory of iterative algorithms. Consider the DC program of:
\begin{equation} \label{}
  \begin{aligned}
    &\underset{x\in{\mathcal{R}^n}}{\text{Minimize}}
& &u_0(x)\\
 & \text{subject to} 
  & & \mbox{}  \  u_i(x)-v_i(x)\leq{0}  \\
  \end{aligned}
\end{equation}
According to \cite{convg} the following conditions are sufficient for convergence of CCP algorithm for any chosen initial point:
\begin{itemize}
\item
$u_i$ and $v_i$ are real-valued differentiable convex functions on $\mathcal{R}^n$.
\item
$\Omega:=\{x : u_i(x)-v_i(x)\leq{0}, i\in\}$ is compact and therefore the point-to-set map related to CCP algorithm is uniformly compact.
\end{itemize}
The first condition is explicitly satisfied for problem \eqref{PP}. Thus we only need to show the compactness of the feasible region. The set is compact if and only if it is closed and bounded. $\Omega$ is closed since all $u_i$s and $v_i$s in problem \eqref{PP} are lower semi-continuous functions and thus as shown in \cite{close} their level sets defined by $L_{\gamma}=\{x\in{\mathcal{R}^n}: u_i(x)-v_i(x)<\gamma\}$ will be closed. Also since $\boldsymbol{X_{lj}}, \forall{l,j}$ are positive definite and their trace is bounded with power constraint it can be inferred that $R_i,E_i,I_i, \forall{i}$ and therefore $\Omega$ is bounded. Thus the second condition is also satisfied in problem \eqref{PP} and the convergence of the algorithm will be proved. 

\section{Proof of proposition \ref{prop3}}
The proposition can be proved by analysing the KKT conditions for the SDP problem \eqref{Q1}. 
Lagrangian for this problem is defined as:
 \begin{equation} \label{L}
  \begin{aligned}
&\mathcal{L}(\gamma_i^1,\gamma_i^2,\gamma_i^3,\gamma_i^4,\gamma^5,\gamma^6,\gamma_i^7,\gamma_i^8,\gamma_i^9, \Gamma_{lj})=\\
&\lambda-\sum_{i=1}^{N_R}{\gamma_i^1(\lambda{v_i^{(1)}}-\alpha_iR_i)}-\sum_{i=1}^{N_{UE}}{\gamma_i^2(\lambda{v_i^{(2)}}-{\eta_i}\beta_iE_i)}-\\
&\sum_{i=1}^{N_{UE}}{\gamma_i^3(R_i-\log(E_i+{\sigma_i^2})+\log(I_i+{\sigma_i^2}))}-\\
&\sum_{i=1}^{N_{UE}}{\gamma_i^4(E_i-\sum_{j=1}^{N_{AP}}\sum_{l=1}^{N_{UE}}\text{trace}(\boldsymbol{H_{ij}}\boldsymbol{X_{lj}})})-\\
&\sum_{i=1}^{N_{UE}}{\gamma_i^5(I_i-\sum_{j=1}^{N_{AP}}\sum_{l=1,l\neq{i}}^{N_{UE}}\text{trace}(\boldsymbol{H_{ij}}\boldsymbol{X_{lj}}))}-\\
&{\gamma^6(\sum_{j=1}^{N_{AP}}\sum_{l=1}^{N_{UE}}{\text{trace}({\boldsymbol{X_{lj}}})}-{P_{max}}})-\\
&\sum_{j=1}^{N_{AP}}\sum_{l=1}^{N_{UE}}\text{trace}(\boldsymbol{\Gamma_{lj}}\boldsymbol{X_{lj}}),
  \end{aligned}
\end{equation}
where $\gamma_i^1,\gamma_i^2,\gamma^6\geq{0}$ and $\boldsymbol{\Gamma_{lj}}\succeq{0}$ are the Lagrangian multipliers. 
The useful KKT conditions for this problem can be written as:
  \begin{align}
  \label{B1}
&\mathop{\nabla}\limits_{\lambda}{\mathcal{L}}=1-\sum_{i=1}^{N_{UE}}(\gamma_i^1{v_i^{(1)}}+\gamma_i^2{v_i^{(2)}})=0\\ 
\label{B2}
&\gamma_i^1(\lambda{v_i^{(1)}}-\alpha_iR_i)=0\\ 
\label{B3}
&\gamma_i^2(\lambda{v_i^{(1)}}-{\eta_i}\beta_iE_i)=0\\ 
\label{B4}
&\mathop{\nabla}\limits_{R_i}{\mathcal{L}v_1^{(2)}}={{\alpha_i}\gamma_i^1}-\gamma_i^3=0\\ 
\label{B5}
&\mathop{\nabla}\limits_{E_i}{\mathcal{L}}={\gamma_i^2}{\eta_i}{\beta_i}+\frac{\gamma_i^3}{E_i+{\sigma_i^2}}-\gamma_i^4=0\\
\label{B6}
&\mathop{\nabla}\limits_{I_i}{\mathcal{L}}=-\frac{\gamma_i^3}{I_i+{\sigma_i^2}}-\gamma_i^5=0\\
\label{B7}
&\boldsymbol{\Gamma_{lj}}\boldsymbol{X_{lj}}=0,\\
\label{B8}
&\mathop{\nabla}\limits_{\boldsymbol{X_{lj}}}{\mathcal{L}}={-\gamma^6}\boldsymbol{I_{N_{A_j}}}+\sum_{i=1}^{N_{UE}}{(\gamma_i^4+\gamma_i^5)\boldsymbol{H_{ij}}-{\gamma_l^5}\boldsymbol{H_{lj}}}+\boldsymbol{\Gamma_{lj}}=0
 \end{align}
Form the slack conditions \eqref{B2}, \eqref{B3} it can be inferred that only one of the $\gamma_i^1, \gamma_i^2, \ \forall{i}$ can be positive. Since if there exists $\hat{i}$ for which we have 
$\gamma_{\hat{i}}^1, \gamma_{\hat{i}}^2>0$, \eqref{B2}, \eqref{B3} result in:
\begin{equation} \label{}
  \begin{aligned}
 \lambda=\frac{\alpha_{\hat{i}}R_{\hat{i}}}{v_{\hat{i}}^1}=\frac{\eta_{\hat{i}}\beta_{\hat{i}}E_{\hat{i}}}{v_{\hat{i}}^2}
 \end{aligned}
\end{equation}
which is not possible in general for fixed values of $v_i^{(2)},v_i^{(2)}$ and $\alpha_i,\beta_i$. Therefore the optimal solution occurs when we have:
\begin{equation} \label{}
  \begin{aligned}
 \lambda=\min_i(\frac{\alpha_{i}R_{{i}}}{v_{{i}}^1},\frac{\eta_i\beta_iE_{{i}}}{v_{{i}}^2})
 \end{aligned}
\end{equation}
which results in:
\begin{equation} 
  \begin{aligned}
\gamma^*= \left\{
  \begin{array}{lr}
    \gamma_{i^*}^1  & \text{if} \ \   \lambda=\frac{\alpha_{i^*}R_{{i}^*}}{v_{{i}^*}^1}\\
   \gamma_{i^*}^2 & \text{if}  \ \  \lambda=\frac{\eta_{i^*}\beta_{i^*}E_{{i}^*}}{v_{{i}^*}^2}
  \end{array}
\right.
  \end{aligned}
\end{equation}
So we have $\gamma_i^*>{0}$ and the other $\gamma_i^1=\gamma_i^2=0, \forall{i}\neq{i^*}$. Now according to \eqref{B4},\eqref{B5},\eqref{B6} we will have 
$\sum_{i=1}^{N_{UE}}{{(\gamma_i^4+\gamma_i^5)\boldsymbol{H_{ij}}}}=a(\gamma^*)\boldsymbol{H_{i^*j}}$ for  the optimal solution in which 
\begin{equation} \label{a}
  \begin{aligned}
a(\gamma^*)= \left\{
  \begin{array}{lr}
   \gamma_{i^*}^1\alpha_{i^*}(\frac{1}{E_{i^*}+{\sigma_{i^*}^2}}-\frac{1}{I_{i^*}+{\sigma_{i^*}^2}}) & \text{if} \ \  \gamma^*=\gamma_{i^*}^1\\
   \gamma_{i^*}^2\eta_{i^*}\beta_{i^*} & \text{if}  \ \  \gamma^*=\gamma_{i^*}^2
  \end{array}
\right.
  \end{aligned}
\end{equation}
Now using \eqref{B8} we have:
\begin{equation}
  \begin{aligned}
\boldsymbol{\Gamma_{lj}}&={\gamma^6}\boldsymbol{I_{N_{A_j}}}-\sum_{i=1}^{N_{UE}}(\gamma_i^4+\gamma_i^5)\boldsymbol{H_{ij}}+{\gamma_l^5}\boldsymbol{H_{lj}}
\\&=
{\gamma^6}\boldsymbol{I_{N_{A_j}}}-a(\gamma^*)\boldsymbol{H_{i^*j}}+{\gamma_l^5}\boldsymbol{H_{lj}}\\
&= {\gamma^6}\boldsymbol{I_{N_{A_j}}}-b_l(\gamma^*)\boldsymbol{H_{i^*j}}
  \end{aligned}
\end{equation}
in which
\begin{equation}
  \begin{aligned}
  b_l(\gamma^*)=
\left\{\begin{array}{lr}
  a(\gamma^*)-{\gamma_l^5}  & l=i^*\\
 a(\gamma^*) & \forall l\neq{i^*}
  \end{array}\right..
  \end{aligned}
\end{equation}
 As a consequence we get:
\begin{equation}
  \begin{aligned}
\text{Rank}(\boldsymbol{\Gamma_{lj}})&\geq{\text{Rank}({\gamma^6}\boldsymbol{I_{N_{A_j}}})-\text{Rank}(b_l(\gamma^*)\boldsymbol{H_{i^*j}})
}\\
& \geq{N_{A_j}-1}, \forall{l,j}.
  \end{aligned}
\end{equation}
On the other side, according to \eqref{B7}  if $\text{Rank}(\boldsymbol{\Gamma_{lj}})=N_{A_j}, \forall{l,j}$  we have $\boldsymbol{X_{lj}}=0, \forall{l,j}$ which is not the optimal solution. Therefore $\text{Rank}(\boldsymbol{\Gamma_{lj}})=N_{A_j}-1, \forall{l,j}$ and rank-null theorem leads to $\text{Rank}(\boldsymbol{X_{lj}})=1, \forall{l,j}$.


\section*{Acknowledgment}
The authors would like to thank IAP BESTCOM project funded by BELSPO, and the FNRS for the financial support.

\ifCLASSOPTIONcaptionsoff
  \newpage
\fi



%

\end{document}